\documentclass[11pt]{article}
\usepackage[dvipsnames,svgnames,x11names,hyperref]{xcolor}
\usepackage{amsmath}
\usepackage{amsthm}
\usepackage[T1]{fontenc}

\usepackage{fullpage}
\usepackage{thmtools}
\usepackage[colorlinks]{hyperref}
\hypersetup{
  linkcolor=[rgb]{0,0,0.4},
  citecolor=[rgb]{0, 0.4, 0},
  urlcolor=[rgb]{0.6, 0, 0}
}
\usepackage{amssymb,amsfonts,amsmath,amsthm,amscd,dsfont,mathrsfs,bbm}
\usepackage{todonotes}
\usepackage{complexity}
\usepackage{tikz}
\usetikzlibrary{decorations.pathreplacing,backgrounds}
\usepackage{multirow}

\usepackage{amsthm}
\usepackage{thmtools,thm-restate}

\numberwithin{equation}{section}
\declaretheoremstyle[bodyfont=\it,qed=\qedsymbol]{noproofstyle}

\declaretheorem[name=Observation,numbered=no]{observation*}

\declaretheorem[numberlike=equation]{fact}

\declaretheorem[numberlike=equation]{theorem}

\declaretheorem[name=Theorem,numbered=no]{theorem*}

\declaretheorem[numberlike=equation]{lemma}
\declaretheorem[name=Lemma,numbered=no]{lemma*}

\declaretheorem[name=Corollary,numbered=no]{corollary*}

\declaretheorem[name=Proposition,numbered=no]{proposition*}

\declaretheorem[name=Claim,numbered=no]{claim*}

\declaretheorem[name=Conjecture,numbered=no]{conjecture*}

\declaretheorem[name=Question,numbered=no]{question*}

\declaretheoremstyle[bodyfont=\it,qed=$\lozenge$]{defstyle} 

\declaretheorem[numberlike=equation,style=defstyle]{definition}
\declaretheorem[unnumbered,name=Definition,style=defstyle]{definition*}

\declaretheorem[unnumbered,name=Example,style=defstyle]{example*}

\declaretheorem[unnumbered,name=Notation=defstyle]{notation*}

\declaretheorem[unnumbered,name=Construction,style=defstyle]{construction*}

\declaretheorem[numberlike=equation,style=defstyle]{remark}
\declaretheorem[unnumbered,name=Remark,style=defstyle]{remark*}

\usepackage{nth}
\usepackage{intcalc}
\usepackage{etoolbox}
\usepackage{xstring}

\usepackage{ifpdf}
\ifpdf
\else
\usepackage[quadpoints=false]{hypdvips}
\fi

\newcommand{\ECCCurl}[2]{http://eccc.hpi-web.de/report/\ifnumcomp{#1}{>}{93}{19}{20}#1/#2/}

\newcommand{\shortECCC}[2]{\texttt{\href{http://eccc.hpi-web.de/report/\ifnumcomp{#1}{>}{93}{19}{20}#1/#2/}{eccc:TR#1-#2}}}

\newcommand{\parseECCC}[1]{
\StrSubstitute{#1}{TR}{}[\tmpstring]%
\IfSubStr{\tmpstring}{/}{ 
\StrBefore{\tmpstring}{/}[\ecccyear]%
\StrBehind{\tmpstring}{/}[\ecccreport]%
}{
\StrBefore{\tmpstring}{-}[\ecccyear]%
\StrBehind{\tmpstring}{-}[\ecccreport]%
}%
\shortECCC{\ecccyear}{\ecccreport}}

\newcommand{\F}{\mathbb{F}}
\newcommand{\N}{\mathbb{N}}
\newcommand{\CC}{\mathbb{C}}

\newcommand{\va}{\mathbf{a}}

\newcommand{\vx}{\mathbf{x}}

\newcommand{\vzero}{\mathbf{0}}

\newcommand{\Perm}{\mathsf{Perm}}
\newcommand{\Det}{\mathsf{Det}}

\newcommand{\rank}{\mathsf{rank}}

\DeclareMathOperator{\Sing}{Sing}
\DeclareMathOperator{\Img}{Im}
\DeclareMathOperator{\adj}{adj}
\DeclareMathOperator{\dc}{\mathsf{dc}}
\DeclareMathOperator{\edc}{\mathsf{edc}}

\newcommand{\set}[1]{\left\{ #1 \right\}}

\title{A Lower Bound on Determinantal Complexity}
\author{Mrinal Kumar\thanks{Department of Computer Science \& Engineering, IIT Bombay. Email:\texttt{ mrinal@cse.iitb.ac.in}}
\and
 Ben Lee Volk\thanks{Efi Arazi School of Computer Science, Reichman University, Israel. Part of this work was done while at the Department of Computer Science, University of Texas at Austin, USA, supported by NSF Grant CCF-1705028, and while at the Center for the Mathematics of Information, California Institute of Technology, USA. Email:\texttt{ benleevolk@gmail.com} }}

\date{}

\begin{document}

\maketitle

\abstract{
The determinantal complexity of a polynomial $P \in \F[x_1,  \ldots, x_n]$ over a field $\F$ is the dimension of the smallest matrix $M$ whose entries are affine functions in $\F[x_1,  \ldots, x_n]$ such that $P = \Det(M)$. We prove that the determinantal complexity of the polynomial $\sum_{i = 1}^n x_i^n$ is at least $1.5n - 3$. 

For every $n$-variate polynomial of degree $d$, the determinantal complexity is trivially at least $d$, and it is a long standing open problem to prove a lower bound which is super linear in $\max\{n,d\}$. Our result is the first lower bound for any explicit polynomial which is bigger by a constant factor than $\max\{n,d\}$, and improves upon the prior best bound of $n + 1$, proved by Alper, Bogart and Velasco \cite{ABV17} for the same polynomial.

\section{Introduction}

\subsection{Computing with Determinants}
The \emph{determinantal complexity} of a polynomial $f \in \F[x_1, \ldots, x_n]$, denoted $\dc(f)$, is the minimal integer $m$ such that there exists an affine map $L : \F^n \to \F^{m \times m}$ such that $f(\vx) = \mathsf{Det}(L(\vx))$, where for every square matrix $M$, $\Det(M)$ denotes the determinant of $M$.

This notion was first implicitly defined by Valiant \cite{V79}, and it is tightly related to the $\VP$ vs.\ $\VNP$ problem, the algebraic analog of the $\P$ vs. $\NP$ problem. The essence of the $\VP$ vs.\ $\VNP$ problem is showing that some explicit polynomials are hard to compute. By defining natural notions of reductions and completeness, Valiant showed that this problem is in fact equivalent to showing that, for fields of characteristic different than two, the determinantal complexity of the permanent polynomial,
\[
\Perm_n(X)  = \sum_{\sigma \in S_n} \prod_{i=1}^n x_{i,\sigma(i)},
\]
doesn't grow like a polynomial function in $n.$\footnote{Strictly speaking, the $\VP$ vs. $\VNP$ question is equivalent to showing that the determinantal complexity of the $\Perm_n$ is at least $n^{\omega(\log n)}$, but we skip over this fine grained detail for now.}

This fact is a consequence of the \emph{completeness} property of the determinant: Valiant showed that if $f$ has an algebraic formula of size $s$, then the determinantal complexity of $f$ is at most $s$. This remains true even if $f$ has an \emph{algebraic branching program} (ABP) of size $s$: ABPs are a natural and more powerful model of computation than formulas. We refer to \cite{SY10} and \cite{S15} for more background on algebraic complexity theory and for proofs of these statements.

Thus, Valiant also established en passant the non-obvious fact that the determinantal complexity of every polynomial is finite, and it's at most roughly $\binom{n+d}{n}$ for every $n$-variate polynomial of degree $d$. Standard counting arguments also show that this estimate is close to being tight for almost every such polynomial.

The benefit of this reformulation of the $\VP$ vs.\ $\VNP$ problem is that it appears to strip away altogether the notion of ``computation'': indeed, this problem can be stated without even defining a computational model in any standard sense of the word, and thus it can potentially be proved without having to argue about the topology or structure of every possible arithmetic computation.

In practice, however, proving lower bounds on determinantal complexity is (unsurprisingly) difficult. Currently, for $n$-variate polynomials, there are no known lower bounds which are super-linear in $n$ (see \autoref{sec:previous results} for more details on previous work). Due to the completeness property mentioned above, a lower bound of $s$ on the determinantal complexity of $f$ will imply the same lower bound for algebraic formulas and even algebraic branching programs. However, super-linear lower bounds for formulas are well-known for decades \cite{K85}, and super-linear lower bounds for ABPs were recently established in \cite{CKSV19}, so there doesn't seem to be any major complexity-theoretic barrier for proving such lower bounds for determinantal complexity: the main obstacle is seemingly lack of techniques for reasoning about computations using determinants, and hence it is important to study this model and to develop techniques to understand it and to prove lower bounds, for the permanent as well as for other explicit polynomials.

Even for the purpose of separating $\VP$ and $\VNP$, one need not necessarily prove a lower bound on the determinantal complexity of the permanent;  the same conclusion will hold if the lower bound on determinantal complexity is shown for any ``explicit'' polynomial (formally, in the class $\VNP$, which we don't define here) in lieu of the permanent.

Before we describe the previous work concerning determinantal complexity, we provide a brief remark about the notion of a ``trivial'' lower bound in this context which is worth remembering when evaluating the previous results (and our result). Unlike most standard computational models, observe that for an $n$-variate polynomial of degree $d$, even a lower bound of $n$ is non-trivial for determinantal complexity. This is because every coordinate of the affine  map $L$ can depend on all $n$ variables. Nevertheless, since the determinant of an $m \times m$ matrix is a degree $m$ polynomial, and thus $\Det(L(\vx))$ is a polynomial of degree at most $m$ for every affine map $L$, the degree $d$ \emph{is} a trivial lower bound on the determinantal complexity of $f$. Therefore, it is natural to consider polynomial families in which $d \le n$ or alternatively to hope to prove lower bounds stronger than $\max\{n,d\}$.

\subsection{Previous work}\label{sec:previous results}

The early work on determinantal complexity mostly focused on proving lower bounds for the permanent. Recall that the $n \times n$ permanent, $\Perm_n$, is a degree $n$ polynomial, so the trivial lower bound is $\dc(\Perm_n) \ge n$. Since over characteristic 2 the permanent and determinant coincide, the results described here hold for characteristic not equal to 2.

Already in 1913, Szeg\H{o} \cite{Szego13}, answering a question of P\'{o}lya \cite{Polya13}, showed that there's no way to generalize the $2 \times 2$ identity
\[
\Perm
\begin{pmatrix}
x_{1,1} & x_{1,2} \\
x_{2,1} & x_{2,2}
\end{pmatrix}
=
\Det
\begin{pmatrix}
x_{1,1} & x_{1,2} \\
-x_{2,1} & x_{2,2}
\end{pmatrix}
\]
by affixing $\pm$ signs to an $n\times n$ matrix of variables for $n \ge 3$.

Marcus and Minc \cite{MM61} strengthened this result by showing that for every $n$, $\dc(\Perm_n) > n$. Subsequent work by von zur Gathen \cite{vzG87}, Babai and Seress (see \cite{vzG87}), Cai \cite{Cai90} and Meshulam \cite{Meshulam} obtained the slightly stronger lower bound $\dc(\Perm_n) \ge \sqrt{2}n$.

Mignon and Ressayre \cite{MR04} greatly improved the lower bound by proving $\dc(\Perm_n) \ge n^2/2$, over the complex numbers. Cai, Chen and Li \cite{CCL10} extended this lower bound to fields of positive characteristic different than two, and Landsberg, Manivel and Ressayre \cite{LMR13} extended this result to the \emph{border} version of determinantal complexity, that is, they showed that the permanent is not even in the closure of polynomials with determinantal complexity less than $n^2/2$. Finally, Yabe \cite{Yabe} obtained an improved lower bound of $(n-1)^2+1$ over the real numbers.

However, while these lower bounds are quadratic in the degree, $\Perm_n$ is a polynomial with $n^2$ many variables, and notably none of these lower bounds is larger than the number of variables. In particular, these results don't even recover a weak form of the $n^3$ formula lower bound of Kalorkoti for $\Perm_n$ \cite{K85}.

Landsberg and Ressayre \cite{LR17} considered determinantal representations that respect certain symmetries (which they called \emph{equivariant determinantal complexity} and denoted $\edc$), and proved that $\edc(\Perm_n)$ is exponential in $n$. It's unclear how stringent the symmetry requirement is; Ladnsberg and Ressayre put forward the ambitious conjecture that $\edc$ and $\dc$ are polynomially related, which, if true, would imply $\VP \neq \VNP$. To the best of our knowledge, this conjecture remains open, but it's worth mentioning that in the context of \emph{regular determinantal complexity}, another notion defined and studied by \cite{LR17}, it can be shown unconditionally that requiring symmetry may result in a super-polynomial blow-up \cite{IL17}.

The question of lower bounds for other explicit polynomial was also considered: Mignon and Ressayre \cite{MR04} proved that the determinantal complexity of quadratic polynomials of rank $r$ is \emph{exactly} $\lceil (r+1)/2 \rceil$ (this, of course, cannot give a lower bound beyond $\lceil (n +1)/2 \rceil$). Chen, Kayal and Wigderson \cite{CKW11} observed that the technique of Mignon and Ressayre implies an $n/2$ lower bound on the determinantal complexity of the elementary symmetric polynomial of degree 2, $\sum_{1 \le i < j \le n} x_i x_j$. Kumar \cite{K19} used a different technique to prove a similar lower bound for the power symmetric polynomials $\sum_i x_i^d$ for $d \ge 2$ over $\CC$.

The last lower bound was improved in a recent work of Alper, Bogart and Velasco \cite{ABV17}: an immediate corollary of their main theorem is that $\dc \left( \sum_i x_i^d \right) \ge n+1$, for every $d \ge 2$. Note that this lower bound is (only slightly) larger than the number of variables $n$, which is the first lower bound we are aware of with this feature. The results of Alper et al.\ are more general, and are stated as a function of the co-dimension of the singular locus of the polynomial, a notion we use as well (see \autoref{sec:main-lb}). In particular they are able to prove that $\dc(\Perm_3)=7$, but their main statement can't imply any lower bound stronger than $n+1$ for an $n$-variate polynomial.

\subsection{Our result}\label{sec:our result}
Our main result is the following theorem. 
\begin{theorem}\label{thm:main}
For every natural number $n \ge 6$, the determinantal complexity of the polynomial  $\sum_{i = 1}^n x_i^n$ over the field of complex numbers  is at least $1.5n - 3$. 
\end{theorem}

Although for simplicity we state our results for the complex numbers, all the results in this paper also hold for algebraically closed fields of positive characteristic $p$, as long as $p$ doesn't divide $n$. This assumption is not only an artifact of the proof. For example, when $n=p^k$, and over characteristic $p$,
\[
\sum_{i=1}^{p^k} x_i^{p^k} = \left( \sum_{i=1}^{p^k} x_i \right)^{p^k}
\]
has determinantal complexity at most $n=p^k$; it is also a polynomial of degree $n$, so its determinantal complexity is at least, and hence equals, $n$.

As discussed in \autoref{sec:previous results}, this is the first non-trivial\footnote{This means that the degree of the polynomials is at most the number of variables.} lower bound of the form $(1 + \epsilon) n$, for any $\epsilon > 0$ for any explicit $n$ variate polynomial family, and improves the previous best bound of $n + 1$ by Alper, Bogart and Velasco \cite{ABV17} by a constant factor.

This result, of course, is not fully satisfactory. The best upper bound we're aware of for $\dc(\sum_{i=1}^n x_i^n)$ is $O(n^2)$, which follows from converting the natural algebraic formula or ABP computing this polynomial to a determinantal expression. We suspect that the true complexity might be $\Omega(n^2)$ or at the very least $\omega(n)$.

Quantitatively, the situation here is somewhat similar to the case of lower bounds on the rank of 3-dimensional tensors, where the best lower bounds are only a constant factor away from the trivial lower bound, and proving super-linear lower bounds remains a challenging open problem (cf.\ \cite{AFT11,BD80,Blaser99,Shpilka01}, among others).

We now give an outline of the main ideas in our proof. 
\subsection{Overview of the proof}
Let $M \in \F[x_1, x_2, \ldots, x_n]^{m \times m}$ matrix of affine functions such that $\sum_{i = 1}^n x_i^n = \Det(M(\vx))$.  \autoref{thm:main} shows a lower bound of $1.5n - 3$ on $m$. There are essentially three main ingredients to the proof of \autoref{thm:main}, and we now discuss them in some more detail. 
\subsubsection*{Converting the matrix $M$ into a normal form}
 Let $M_0 \in \F^{m \times m}$ be the \emph{constant part} of the matrix $M$, i.e. $M_0 = M(\mathbf{0})$.  As a first step of our proof, we show (in \autoref{lem: rearranging constant terms}) that without loss of generality, $M_0$ can be assumed to be a diagonal matrix of rank equal to  $m-1$.  We a say that a matrix $M$ is in \emph{normal form} if it has this additional structure. 
 
It is quite easy to observe that the rank of $M_0$ is at most $m-1$. However, for technical reasons, we actually need the lower bound on the rank as well, and this fact is a consequence of  comparing the dimensions (as algebraic varieties)  of the singular locus (which is just the the set of  zeroes of a polynomial of multiplicity at least two) of the determinant and that of the polynomial $\sum_{i = 1}^n x_i^n$. Observations of this nature have been used in the context of determinantal complexity lower bounds before, and indeed, we crucially rely on a well known lemma of von~zur~Gathen (see \autoref{fact:singular locus of det}) for the proof. The details can be found in \autoref{sec: redn to normal form}. 
 
\subsubsection*{Determinantal complexity of higher degree polynomial maps}
As the key ingredient of our proof, we show that for any matrix $M(\vx) \in \F[\vx]^{m\times m}$ where the entries of $M$ are polynomials  of degree at most $n-1$ and $M$ is in normal form, if $\Det(M(\vx)) =\sum_{i = 1}^n x_i^n $, then $m \geq n/2$. Moreover, roughly the same lower bound continues to hold  as long as $\det(M) = \left(\sum_{i = 1}^n x_i^n\right)(1 + Q)$ for any polynomial $Q$, with $Q(\mathbf{0}) = 0$. 

Thus, this is a significant generalization of the $n/2$ lower bound on the standard notion determinantal complexity (where the entries of $M$ are affine functions) of $\sum_{i = 1}^n x_i^n$ as shown in \cite{K19}: this shows that roughly the same lower bound continues to hold even when the entries of the matrix are arbitrary polynomials of degree as high as $n-1$ and the determinant of the matrix equals an arbitrary multiple of $\sum_{i = 1}^n x_i^n$ with a non-zero constant term. 

The proof of the lemma relies on the observation that the polynomial $\sum_{i = 1}^n x_i^n$ does not vanish with multiplicity at least two very often. This seemingly simple observation has been previously used in the context of lower bounds on algebraic branching programs computing this polynomial \cite{K19, CKSV19} in a crucial way. See \autoref{sec:lb for higher deg maps} for further details.

\subsection*{Trading dimension of the matrix for degree}
As the  final ingredient of our proof, we use  a well known property of determinants (\autoref{lem:reducing size}) to show that if there is an $m \times m$ matrix $M$ whose entries are affine functions and $\Det(M) = \sum_{i = 1}^n x_i^n$, then there is an  $(m - n + 2)\times (m-n + 2)$ matrix $N$  whose entries are polynomials of degree at most $n-1$ and $\Det(N) = (\sum_{i = 1}^n x_i^n)(1 + Q)$ for a polynomial $Q$ which vanishes at zero. Moreover, if the matrix $M$ is in normal form, then the matrix $N$ continues to be in normal form. 

Thus, we are in a setup where we can invoke the lower bound in \autoref{lem: n/2 lb for higher deg entries} discussed earlier and we get that the dimension of $N$ which equals $m - n + 2$ must be at least $n/2-1$, thereby implying that $m$ is at least $1.5n - 3$. The details of this step can be found in \autoref{sec:completing the proof}. 

\section{Preliminaries}\label{sec:prelim}
In this paper $\F$ always denotes an algebraically closed field. We use $\vx$ to denote a tuple of $n$ variables $x_1, \ldots, x_n$, where $n$ is understood from the context (or is otherwise explicitly mentioned).

We consider polynomial maps $M : \F^n \to \F^{m \times m}$ given by $m^2$ polynomials $(M_{i,j})_{i,j\in [m]}$. The same object can be thought of as a matrix of polynomials $M(\vx) \in \F[\vx]^{m \times m}$ and we use both points of view interchangeably. The degree of $M$ is the maximum degree of its coordinates, i.e., $\deg M = \max_{i,j} \deg M_{i,j}$.

Each $M(\vx) \in \F[\vx]^{m \times m}$ can be uniquely written as $M(\vx) = M'(\vx) + M_0$, where $M_0 \in \F^{m \times m}$ and in all $m^2$ coordinates of $M'$, the constant term is zero. We then call $M_0$ the \emph{constant part} of the map. A polynomial in which the constant term is zero is called \emph{constant free}, and a polynomial map is called constant free if all of its coordinates are constant free, i.e., in the above decomposition, $M_0=0$.

We denote the determinant polynomial by $\Det$. In cases where it is important to emphasize the dimension of the matrices in question we write it in the subscript, so for example the $m \times m$ determinant polynomial is denoted by $\Det_m$.

We assume knowledge of basic concepts in algebraic geometry such as affine varieties $V \subseteq \CC^n$ and their dimension, which we denote $\dim(V)$. We encourage readers unfamiliar with those terms to consult the excellent textbook \cite{CLO07}. 

\subsection*{Determinantal Complexity}
We now formally define the notion of determinantal complexity, which is the focus of this paper. 
\begin{definition}[Determinantal Complexity]
The determinantal complexity of a polynomial $P \in \F[\vx]$ is defined as the minimum $m \in \N$ such that there is a $m \times m$ matrix $M \in \F[\vx]$ whose entries are polynomials of degree at most one such that 
\[
P = \Det(M) \, . \qedhere
\]
\end{definition}
\begin{remark}
The above definition naturally generalizes to a family of polynomials in the following sense. A family  $\{P_n\}_{n\in \N}$ of polynomials is said to have determinantal complexity at most $f(n): \N \to \N$ if there exists an $n_0 \in \N$, such that for every $n \geq n_0$, the determinantal complexity of $P_n$ is at most $f(n)$. 
\end{remark}

\section{A lower bound on determinantal complexity}\label{sec:main-lb}

This section will be devoted for a proof of \autoref{thm:main}. We begin with the following lemma, which was instrumental in the recent proofs of lower bounds for algebraic formulas and algebraic branching programs.

\begin{lemma}[\cite{CKSV19, K19}]\label{lemma: width lb}
Let $d \ge 2$ be a natural number. Let $P_1, P_2, \ldots, P_t, Q_1, \ldots, Q_t, \allowbreak L \in \CC[\vx]$ be polynomials such that $\deg(P') < d$, $P_1, \ldots, P_t, Q_1, \ldots, Q_t$ have a common zero and 
\[
\sum_{i = 1}^n x_i^d = \sum_{j = 1}^t P_j(\vx)Q_j(\vx) + P' \, .
\]
Then, $t \geq n/2$. 
\end{lemma}

We now show that without loss of generality, the constant part of every polynomial map $M$ such that $\sum_{i=1}^n x_i^d = \Det_m(M(\vx))$ has a very special form: is it an $m \times m$ diagonal matrix with $0$ in the $(1,1)$ coordinate and $1$ in all diagonal entries.

\subsection{Reducing the matrix \texorpdfstring{$M$}{M} to a normal form}\label{sec: redn to normal form}

This claim is not entirely new and very similar statements were proved, for example, in \cite{MR04, ABV17}. For completeness, and since the exact statement we need is slightly more general, we provide a proof.

\begin{lemma}\label{lem: rearranging constant terms}
Let $d \ge 2$ be a natural number and let $M(\vx) \in \F[\vx]^{m \times m}$ be a polynomial map such that 
\[
\Det_m(M(\vx)) = \sum_{i = 1}^n x_i^d \, .
\]
Then, there exists a matrix $\tilde{M}(\vx)\in \F[\vx]^{m \times m}$ with $\deg(\tilde{M}) \leq \deg(M)$, 
\[
\Det_m(\tilde{M}(\vx)) = \sum_{i = 1}^n x_i^d \, ,
\]
and the constant part of $\tilde{M}$ is a diagonal $m \times m$ matrix $\tilde{M}_0$ such that $(\tilde{M}_0)_{1,1}=0$ and $(\tilde{M}_0)_{i,i}=1$, for $2 \le i \le m$.
\end{lemma}

To prove \autoref{lem: rearranging constant terms} we require a few preliminaries. We begin with the definition of a singular locus of a polynomial (or a hypersurface).

\begin{definition}\label{def:singular locus}
Let $f \in \F[\vx]$ be a polynomial. The \emph{singular locus} of $f$, denoted $\Sing(f)$, is the variety defined by
\[
\Sing(f) = \set{\va : \frac{\partial f }{\partial x_i} (\va) = 0, 1 \le i \le n}. \qedhere
\]
\end{definition}

The singular locus of the determinant was studied by von zur Gathen, who proved the following fact.

\begin{fact}[\cite{vzG87}]\label{fact:singular locus of det}
Let $\F$ be an algebraically closed field and let $\Det_m $ denote the $m \times m$ determinant polynomial. Then $\Sing(\Det_m) \subseteq \F^{m \times m}$ is precisely the set of matrices of rank at most $m-2$, and $\dim \Sing(\Det_m) = m^2 - 4$.
\end{fact}

The following is a slight generalization of a lemma of von zur Gathen (cf.\ also \cite{ABV17}).

\begin{lemma}\label{lem:image misses singular locus}
Let $f \in \F[\vx]$ be a polynomial, and let $M : \F^n \to \F^{m \times m}$ be a polynomial map such that $f(\vx) = \Det_m(M(\vx))$. Suppose further that $\dim(\Sing(f)) < n-4$. Then $\Img(M) \cap \Sing(\Det_m) = \emptyset$. Furthermore, all matrices in $\Img(M)$ have rank at least $m-1$.
\end{lemma}

\begin{proof}
Let $y_{i,j}$ denote the coordinates of $\F^{m \times m}$ and write $M = (M_{i,j})_{i ,j \in [m]}$. Using the chain rule, we compute
\begin{equation}\label{eq:chain rule}
\frac{\partial f}{\partial x_k} = \sum_{i, j \in [m]} \frac{\partial \Det_m}{\partial y_{i,j}} (M(\vx)) \cdot \frac{\partial M_{i,j}}{\partial x_k} (\vx), \quad k \in [n].
\end{equation}
Suppose $A \in \Img(M) \cap \Sing(\Det_m)$, and let $B$ be such that $A = M(B)$. By definition of $\Sing(\Det_m)$, $\frac{\partial \Det_m}{\partial y_{i,j}} (M(B)) = 0$ for all $i, j \in [m]$, and  by \eqref{eq:chain rule} we get that $B \in \Sing (f)$. Thus $M^{-1} (\Sing(\Det_m)) \subseteq \Sing(f)$, and $\dim (M^{-1} (\Sing(\Det_m))) \le  \dim \Sing(f) < n-4$. On the other hand, using a standard lower bound on the dimension of pre-images of polynomial maps (see Theorem 17.24 of \cite{HarrisAlgebraicGeometry}), if $\Img(M)$ and $\Sing(\Det_m)$ aren't disjoint,
\[
\dim(M^{-1} (\Sing(\Det_m))) \ge n + (m^2 - 4) - m^2 = n-4.
\]
This contradiction implies that $\Img(M) \cap \Sing(\Det_m) = \emptyset$. The ``furthermore'' part of the theorem follows from \autoref{fact:singular locus of det}.
\end{proof}

We will also need the following easy fact which shows that $\sum_{i=1}^n x_i ^d$ satisfies that assumption of \autoref{lem:image misses singular locus}.

\begin{fact}[\cite{K19, CKSV19}]\label{fact:sing locus of power-sym}
For every $d \ge 2$, $\dim(\Sing(\sum_{i=1}^n x_i^d)) = 0$.
\end{fact}

We are now ready to prove \autoref{lem: rearranging constant terms}.

\begin{proof}[Proof of \autoref{lem: rearranging constant terms}]
Let $f = \sum_{i=1}^n x_i^d$ and let $M : \F^n \to \F^{m \times m}$ be a polynomial map such that $f(\vx)  = \Det_m(M(\vx))$, and write $M=M'+M_0$ where $M_0$ is the constant part of $M$.

First, observe that
\[
0 = f(\vzero) = \Det_m(M(\vzero)) = \Det_m(M_0),
\]
which implies that $\rank(M_0) < m$. By \autoref{lem:image misses singular locus} and \autoref{fact:sing locus of power-sym}, we also know that $\rank(M_0) = \rank(M(\vzero)) \ge m-1$, so $\rank(M_0) = m-1$.

By performing Gaussian elimination on the rows and on the columns, we can find two $m \times m$ matrices $G_1, G_2$ such that $\det(G_i) = \pm 1$ for $i=1,2$ and $N_0 := G_1 M_0 G_2$ is a diagonal matrix such that $(N_0)_{1,1} = 0$ and $(N_0)_{i,i} \neq 0$ for $2 \le i \le m$.

Now define a diagonal $m\times m$ matrix $\Delta$ such that $\Delta_{i,i} = 1/(N_0)_{i,i}$ for $2 \le i \le m$, and
\[
\Delta_{1,1} = \Det(G_1) \cdot \Det(G_2) \cdot \prod_{i=2}^m (N_0)_{i,i}.
\]
It readily follows that $\Det(\Delta) = \Det(G_1) \cdot \Det(G_2)$, and that $\tilde{M_0} := (G_1 M_0 G_2) \Delta$ is a diagonal matrix such that $(\tilde{M_0})_{1,1} = 0$ and $(\tilde{M_0})_{i,i} = 1$ for all $2 \le i \le m$.

Finally, define $\tilde{M} = G_1 M G_2 \Delta$. We verify that indeed
\begin{align*}
\Det(\tilde{M}(\vx)) &= \Det(G_1) \cdot \Det(M(\vx)) \cdot \Det(G_2) \cdot \Det(\Delta) \\
&= \Det(M(\vx)) \cdot (\Det(G_1) \cdot \Det(G_2))^2 = \Det(M(\vx)) = f(\vx).
\end{align*}
We also have that
\[
\tilde{M} = G_1 (M' + M_0) G_2 \Delta = G_1 M' G_2 \Delta + G_1 M_0 G_2 \Delta = G_1 M' G_2 \Delta + \tilde{M_0}.
\]
Since $G_1, G_2, \Delta \in \F^{m \times m}$, it also holds that $\tilde{M}' := G_1 M' G_2 \Delta$ is a matrix of constant-free polynomials, and that  $ \deg \tilde{M} \le  \deg M$. 
\end{proof}

We will also use the following simple and well known property of the determinant of a block matrix.

\begin{lemma}\label{lem:reducing size}
Let $M \in \F^{m \times m}$ be a matrix, and let $A \in \F^{t\times t}, B \in \F^{t\times m-t}, C \in \F^{m-t\times t}, D \in \F^{m-t\times m-t}$ be its submatrices as follows:
\[ M = \begin{pmatrix}
   A & B\\\
C & D
 \end{pmatrix}
 \]
 If $D$ is invertible, then 
 \[
 \Det (M)  = \Det (A - BD^{-1}C ) \cdot \Det (D) \, .
 \]
\end{lemma}

\begin{proof}
Follows directly from the decomposition
\[ \begin{pmatrix}
   A & B\\\
C & D
 \end{pmatrix}
 =
 \begin{pmatrix}
   A - BD^{-1}C & BD^{-1}\\\
0 & I_{m-t}
 \end{pmatrix}
 \cdot 
  \begin{pmatrix}
   I_t & 0\\\
C & D
 \end{pmatrix}
 \]
 and the multiplicativity of the determinant.
\end{proof}

\subsection{Determinantal complexity of higher degree polynomial maps}\label{sec:lb for higher deg maps}

In the following lemma we prove a lower bound of $n/2$ on the determinantal complexity in a more general model than the standard model. This is a generalization with respect to two properties. First, the entries of the matrix are no longer constrained to be polynomials of degree at most $1$, and can have degree as high as $d-1$, while computing the degree $d$ polynomial $\left(\sum_{i = 1}^n x_i^d\right)$. Moreover, the determinant of the matrix $M$ does not even have to compute the candidate hard polynomial $\left(\sum_{i = 1}^n x_i^d\right)$ exactly. It suffices if the determinant is equal to a polynomial of the form $\left(\sum_{i = 1}^n x_i^d\right) \cdot (\beta + Q)$ where $\beta$ is a non-zero field constant and $Q$ is an arbitrary polynomial (of potentially very high degree!) which is constant free, i.e. $Q(\mathbf{0}) = 0$.  

\begin{lemma}\label{lem: n/2 lb for higher deg entries}
Let $d \ge 2$ be a natural number and let $M(\vx) \in \F[\vx]^{m \times m}$ such that $\deg(M) \le d-1$, and the constant part of $M$ is a diagonal matrix $M_0$ such that $(M_0)_{1,1}=0$ and $(M_0)_{i,i}=1$ for $2 \le i \le m$.  Suppose that
\[
\Det(M) = \left(\sum_{i = 1}^n x_i^d\right) \cdot (\beta + Q) \, , 
\]
where $\beta \in \F$ is non-zero and $Q$ is a constant free polynomial. Then $m \geq n/2-1$. 
\end{lemma}
\begin{proof}
 Using the Laplace expansion of $\Det(M)$ along the first row, we get 
\[
\Det(M) = \sum_{j=1}^m (-1)^{(j + 1)}M_{1, j}\cdot \Det (N_{1, j}) \, ,
\]
where $N_{i, j}$ is the submatrix of $M$ obtained by deleting the $i$-th row and the $j$-th column. For every $j \in [m]$, $j > 1$, we claim that $\Det(N_{1, j})$ is a constant free polynomial, i.e. \[
\Det(N_{1, j})(\mathbf{0}) = \Det\left(N_{1, j}(\mathbf{0})\right) = 0 \, .
\] 
To see this, we observe that for every $j\in [m]\setminus \{1\}$, $N_{1, j}(\mathbf{0})$ is a $(m-1)\times (m-1)$ matrix, which has at most $m-2$ non-zero entries. This follows since $M_0$ has at most $m-1$ non-zero entries and in obtaining $N_{1, j}$ from $M$, we drop the entry $M_{j, j}$, which is one of the $(m-1)$ entries of $M$ with a non-zero constant term, and hence one of the $(m-1)$ non-zero entries of $M_0$. However, we note that $N_{1,1}(\vzero)$ is the $(m-1) \times (m-1)$ identity matrix, so the  constant term of $\Det(N_{1,1})$ is $1$, and we write  $\Det(N_{1,1}) = 1 + P(\vx)$ where $P$ is  constant free. Therefore, we have 
\[
\left(\sum_{i = 1}^n x_i^d\right)\cdot (\beta + Q) = \Det(M) = M_{1,1}(1 + P) + \sum_{j = 2}^m (-1)^{(j + 1)}M_{1, j}\cdot \Det(N_{1, j}) \, 
\]
In other words, 
\[
\left(\sum_{i = 1}^n x_i^d\right)\cdot (\beta + Q) = \Det(M) = M_{1,1} +  M_{1,1}\cdot P + \sum_{j = 2}^m (-1)^{(j + 1)}M_{1, j}\cdot \Det(N_{1, j}) \, 
\]
Slightly rearranging (and using $\beta \neq 0$), we get
\[
\sum_{i = 1}^n x_i^d = \frac{1}{\beta}\left(-\left(\sum_{i = 1}^n x_i^d\right)\cdot Q + M_{1,1} +  M_{1,1}\cdot P + \sum_{j = 2}^m (-1)^{(j + 1)}M_{1, j}\cdot \Det(N_{1, j})  \right)\, 
\]
Since, $\deg(M_{1,1}) < d$ and $M_{1,1}, P, M_{1, 2}, \Det(N_{1, 2}), \ldots, M_{1, k}, \Det(N_{1, k}), Q$ are all constant free (and hence share a common zero, namely $\mathbf{0}$), we have from \autoref{lemma: width lb} that $m \geq n/2-1$. 
\end{proof}

\subsection{Completing the proof of Theorem \ref{thm:main}}\label{sec:completing the proof}
We are now ready to complete the proof of \autoref{thm:main}. 
\begin{proof}[Proof of \autoref{thm:main}]
Let $M$ be an $m \times m$ matrix with $\deg(M) \leq 1$ such that 
\[
\sum_{i = 1}^n x_i^n = \Det(M) \, .
\]
From \autoref{lem: rearranging constant terms}, we can assume without loss of generality that the constant part $M_0$ of $M$ is a diagonal matrix such that $(M_0)_{1,1} = 0$ and $(M_0)_{i,i}=1$ for $2 \le i \le m$. In particular, all the off diagonal entries of $M$ and $M_{1,1}$ are homogeneous linear forms or zero, and $M_{j, j} \neq 0$ for $j > 1$. 

Observe that for every $t \leq m-1$, the principal minor $D_t$ of $M$ which is obtained by deleting the first $m -t$ rows and columns of $M$ is invertible over the field of rational functions $\F(\vx)$. To see this, observe that the matrix $D_t(\mathbf{0})$ is the identity matrix, which implies that $\Det(D_t)$ is a non-zero polynomial. Moreover, since every entry of $M$ has degree at most $1$, and $\Det(M)$ has degree $n$, we know that $m \geq n$. So, we conclude that the principal minor $D := D_{(n-2)}$ of $M$ is invertible over $\F(\vx)$. Thus, if $B$ and $C$ are respectively the submatrices of $M$ defined as 
\[ M = \begin{pmatrix}
   A & B\\\
C & D
 \end{pmatrix}
 \]
then by \autoref{lem:reducing size} we have 
 \begin{equation}\label{eq:block determinant}
 \Det(M)  = \Det(A - BD^{-1}C)\cdot \Det(D) \, .
 \end{equation}
Since $D^{-1} = \adj(D)/\det(D)$, where $\adj(D)$ is the adjugate matrix of $D$, the entries of $D^{-1}$ can be written as as a ratio of two polynomials, where the numerator has degree at most $n-3$ and the denominator, which is equal to $\Det(D)$, has degree at most $n-2$. Moreover, as discussed earlier in the proof, 
the constant part of $D$ is the identity matrix, so there is a constant free polynomial $R \in \F[\vx]$
such that 
\[
\Det(D) = 1 + R \, .
\]
Thus, every entry of the $(m-n+2) \times (m-n+2)$ matrix $A - BD^{-1}C$ can be written as a ratio of two polynomials with the numerator being a polynomial of degree at most $n-1$ and the denominator being equal to $\Det(D)= 1 + R$. Therefore, by clearing the denominators and using \eqref{eq:block determinant}, we get that 
 \[
 \Det(M) \cdot (1 + R)^{m - n + 2}  = \Det(N)\cdot (1 + R) \, ,
 \]
 where $N$ is the matrix with polynomial entries of degree at most $n-1$ obtained by multiplying every entry of $A - BD^{-1}C$ by $1 + R$. Simplifying further, we get 
 \[
 \left(\sum_{i=1}^n x_i^n \right) \cdot (1+R)^{m-n+1} = \Det(M) \cdot (1 + R)^{m - n + 1}  = \Det(N)\, .
 \]
 We are almost ready to invoke \autoref{lem: n/2 lb for higher deg entries} to obtain a lower bound on the size of $N$ (and hence $M$), but to do that we need to ensure that the constant part of $N$, $N_0$, is a diagonal matrix with $(N_0)_{1,1}=0$ and $(N_0)_{i,i}=1$ for $2 \le i\le m-n+2$. We now verify that this is indeed the case. 
 
Recall that by the structure of the constant part $M_0$ of $M$, all the entries of $B$ and $C$ and the $(1,1)$ entry of $A$ are constant free, and the constant term of $A_{i,i}$ is $1$ for $2 \le i \le m-n+2$. Thus, every entry of the matrix $BD^{-1}C$ is a rational function with a constant free numerator, and hence all the off-diagonal entries in $A - BD^{-1}C$ as well as its $(1,1)$ entry are rational functions with a constant free numerator. Moreover, the denominator of all the entries $\left(A - BD^{-1}C\right)$  equals $\Det(D) = 1 + R$, for a constant free polynomial $R$. So, expressing each entry of $A-BD^{-1}C$ as a quotient of polynomials, the constant term of each numerator on the diagonal  is $1$ except for the $(1,1)$ entry, which has a constant free numerator.  Finally, observe that eliminating the denominator of the entries of  $\left(A - BD^{-1}C\right)$ by multiplying every entry by $(1 + R)$ gives us the  matrix $N$.

Thus the matrix $N$ satisfies the hypothesis of \autoref{lem: n/2 lb for higher deg entries}, and hence $(m - n + 2) \geq n/2 -1$. This gives us $m \geq 1.5n - 3$ and completes the proof of \autoref{thm:main}.  
\end{proof}

\section*{Acknowledgment}
Mrinal thanks Ramprasad Saptharishi for various discussions on determinantal complexity over the years, and in particular for explaining the proof of the result of Mignon and Ressayre to him. 
\bibliographystyle{customurlbst/alphaurlpp}
\bibliography{references}

\end{document}